\documentclass[a4paper]{article}

\usepackage{xspace}
\usepackage[linesnumbered,ruled]{algorithm2e}
\usepackage{tikz}
\usepackage{fullpage}
\usepackage{amsmath}
\usepackage{amsthm}
\usepackage{graphicx}

\newcommand{\dist}{\ensuremath{\operatorname{dist}\xspace}}

\newcommand\ProblemName{{\sc Minimum Additive $t$-Spanner Problem}}

\theoremstyle{plain}
\newtheorem{theorem}{Theorem}
\newtheorem{lemma}[theorem]{Lemma}
\newtheorem{claim}[theorem]{Claim}
\newtheorem{proposition}[theorem]{Proposition}

\title{An FPT Algorithm for Minimum Additive Spanner Problem}
\author{Yusuke Kobayashi\thanks{Research Institute for Mathematical Sciences, Kyoto University, Japan. Supported by JST ACT-I Grant Number JPMJPR17UB, and JSPS KAKENHI Grant Numbers JP16K16010, 16H03118, and JP18H05291, Japan. Email: yusuke@kurims.kyoto-u.ac.jp}}

\begin{document}

\maketitle

\begin{abstract}
For a positive integer $t$ and a graph $G$, an additive $t$-spanner of $G$ is a spanning subgraph 
in which the distance between every pair of vertices is at most the original distance plus $t$. 
{\sc Minimum Additive $t$-Spanner Problem} is 
to find an additive $t$-spanner with the minimum number of edges in a given graph, 
which is known to be NP-hard. 
Since we need to care about global properties of graphs when we deal with additive $t$-spanners, 
{\sc Minimum Additive $t$-Spanner Problem} is hard to handle, 
and hence only few results are known for it. 
In this paper, 
we study {\sc Minimum Additive $t$-Spanner Problem} from the viewpoint of parameterized complexity. 
We formulate a parameterized version of the problem 
in which the number of removed edges is regarded as a parameter, and give a fixed-parameter algorithm for it. 
We also extend our result to $(\alpha, \beta)$-spanners.  
\end{abstract}

\section{Introduction}

\subsection{Spanners}
\label{sec:spanner}

A {\em spanner} of a graph $G$ is a spanning subgraph of $G$ 
that approximately preserves the distance between every pair of vertices in $G$. 
Spanners were introduced in~\cite{Aw1985,PS1989,PU1989} in the context of synchronization in networks. 
Since then, spanners have been studied with applications to several areas such as
%%turned out to be relevant to several areas such as 
space efficient routing tables~\cite{COWEN200479,PUpfal1989}, 
computation of approximate shortest paths~\cite{Cohen1998,Cohen2000,Elkin2005}, 
distance oracles~\cite{BK2006,Thorup2005}, and so on.

A main topic on spanners is 
trade-offs between the sparsity (i.e., the number of edges) of a spanner and its quality of approximation of the distance,  %% (called {\em stretch}), 
and there are several ways to measure the approximation quality. 
In the early studies, the approximation quality of spanners was measured by a multiplicative factor, 
i.e., the ratio between the distance in the spanner and the original distance. 
Formally, for a positive integer $t$ and a graph $G$, a spanning subgraph $H$ of $G$ is said to be a {\em multiplicative $t$-spanner} 
if $\dist_H(u, v) \le t \cdot \dist_G(u, v)$ holds for any pair of vertices $u$ and $v$. 
Here, $\dist_G(u, v)$ (resp.~$\dist_H(u, v)$) denotes the distance between $u$ and $v$ in $G$ (resp.~in $H$). 
A well-known trade-off between the sparsity and the multiplicative factor is as follows:  
for any positive integer $t$ and any graph $G$, 
there exists a $(2t-1)$-spanner with $O(n^{1+1/t})$ edges~\cite{ADDJS1993}, 
where $n$ denotes the number of vertices in $G$.
%%of the original graph throughout the paper. 
This bound is conjectured to be tight based on the popular Girth Conjecture of Erd\H{o}s~\cite{Erdos1964}. 

Another natural measure of the approximation quality
is the difference between the distance in the spanner and the original distance. 
For a positive integer $t$ and a graph $G$, a spanning subgraph $H$ of $G$ is said to be an {\em additive $t$-spanner} 
if $\dist_H(u, v) \le \dist_G(u, v) + t$ holds for any pair of vertices $u$ and $v$. 
Since an additive spanner was introduced in~\cite{LS1991,LS93}, 
trade-offs between the sparsity and the additive term have been actively studied. 
It is shown in~\cite{ACIM1999,EP2004} that 
every graph has an additive $2$-spanner with $O(n^{3/2})$ edges. %%, where $V$ is the vertex set of the graph. 
In addition, every graph has an additive $4$-spanner with $O(n^{7/5} {\rm poly}(\log n) )$ edges~\cite{Che2013}, and
every graph has an additive $6$-spanner with $O(n^{4/3})$ edges~\cite{BKMP2010}. 
On the negative side, it is shown in~\cite{Abboud2017} that 
these bounds cannot be improved to $O(n^{4/3 - \epsilon})$ for any $\epsilon > 0$.  

As a common generalization of these two concepts, 
$(\alpha, \beta)$-spanners have also been studied in the literature. 
%%By combining these two measures,  
%%we can also consider $(\alpha, \beta)$-spanners. 
For $\alpha \ge 1$, $\beta \ge 0$, and a graph $G$, 
a spanning subgraph $H$ of $G$ is said to be an {\em $(\alpha, \beta)$-spanner} 
if $\dist_H(u, v) \le \alpha \cdot \dist_G(u, v) + \beta$ holds for any pair of vertices $u$ and $v$. 
See~\cite{Bollobas2005,Elkin:2018,Knu2014,Pettie2009,Roditty2005,Thorup2006,Woodruff2006,Woodruff2010} for other results on 
trade-offs between the sparsity of a spanner and its approximation quality.

In this paper, we consider a classical but natural and important problem that finds a spanner of minimum size. 
In particular, we focus on additive $t$-spanners and 
consider the following problem for a positive integer $t$. 

\smallskip

\begin{description}
\item[\ProblemName]
\item[Instance.]
A graph $G=(V, E)$.
\item[Question.]
Find an additive $t$-spanner $H = (V, E_H)$ of $G$ that minimizes $|E_H|$. 
\end{description}

\smallskip

{\sc Minimum Multiplicative $t$-Spanner Problem} 
and {\sc Minimum $(\alpha, \beta)$-Spanner Problem} 
are defined in the same way. 
%%a similar way. 
Such a problem is sometimes called {\sc Sparsest Spanner Problem}. 

Although additive $t$-spanners have attracted attention as described above, 
there are only few results on \ProblemName. 
For any positive integer $t$, 
\ProblemName\ is shown to be NP-hard in~\cite{LS93}. 
Every connected interval graph has an additive $2$-spanner that is a spanning tree~\cite{Kratsch2003},  
which implies that {\sc Minimum Additive $t$-Spanner Problem} in interval graphs can be solved in polynomial time for $t \ge 2$. 
The same result holds for  AT-free graphs~\cite{Kratsch2003}. 
It is shown in~\cite{Chepoi2005} that every chordal graph has an additive $4$-spanner with at most $2n-2$ edges, 
which implies that there exists a $2$-approximation algorithm for {\sc Minimum Additive $4$-Spanner Problem} in chordal graphs. 
To the best of our knowledge, 
no other positive results (e.g., polynomial-time algorithms for special cases or approximation algorithms)
exist for \ProblemName, 
which is in contrast to the fact that 
{\sc Minimum Multiplicative $t$-Spanner Problem} 
has been actively studied 
from the viewpoints of graph classes and approximation algorithms
(see Section~\ref{sec:related}). 

We make a remark on a difference between 
multiplicative $t$-spanners and additive $t$-spanners. 
%%
%%We remark here that 
%%multiplicative $t$-spanners can be characterized 
%%as follows (see~\cite{CK1994,MADANLAL199697,kobayashi2018}): 
As in~\cite{CK1994,MADANLAL199697,kobayashi2018}, 
multiplicative $t$-spanners can be characterized as follows: 
a subgraph $H = (V, E_H)$ of $G=(V, E)$ is a multiplicative $t$-spanner if and only if
$\dist_H (u, v) \le t$ holds for any $uv \in E \setminus E_H$. 
This characterization means that
we only need to care about local properties of graphs
%%we do not need to care global properties of graphs
when we deal with multiplicative $t$-spanners.  
%%This enables us to deal with {\sc Minimum Multiplicative $t$-Spanner Problem}.  
In contrast, for additive $t$-spanners,
no such characterization exists, and hence 
we have to consider global properties of graphs. 
In this sense, handling 
\ProblemName\ is much harder than
{\sc Minimum Multiplicative $t$-Spanner Problem}, 
which is a reason why 
only few results exist for \ProblemName.

\subsection{Our Results}

%%In this paper, we consider a parameterized version of \ProblemName\ 
In this paper, we consider \ProblemName\ from the viewpoint of fixed-parameter tractability
and give a first fixed-parameter algorithm for it. 
A parameterized version of {\sc Minimum Multiplicative $t$-Spanner Problem} 
is studied in~\cite{kobayashi2018}. %%, and a fixed-parameter algorithm for it is presented in the same paper.  
Since an additive (or multiplicative) $t$-spanner of a connected graph contains $\Omega(|V|)$ edges, 
the number of edges of a minimum additive (or multiplicative) $t$-spanner is not an appropriate parameter. 
Therefore, as in~\cite{kobayashi2018}, a parameter is defined as the number of edges that are removed to obtain an additive (or multiplicative) $t$-spanner. 
Note that the same parameterization is also adopted in~\cite{Bang2018} for another network design problem. 
%%That is, we consider the following parameterized problem for fixed $t$. 
Our problem is formulated as follows. 

\smallskip

\begin{description}
\item[{\sc Parameterized} \ProblemName]
\item[Instance.]
A graph $G=(V, E)$.
\item[Parameter.]
A positive integer $k$. 
\item[Question.]
Find an edge set $E' \subseteq E$ with $|E'| \ge k$ such that $H = (V, E \setminus E')$ is an additive $t$-spanner of $G$
or conclude that such $E'$ does not exist. 
%%Find a $t$-spanner $H = (V, E_H)$ of $G$ such that $|E \setminus E_H| \ge k$. 
\end{description}

\smallskip

Note that if there exists a solution of size at least $k$, then its subset of size $k$ is also a solution, 
which means that  we can replace the condition $|E'| \ge k$ with $|E'| = k$ in the problem. 
In this paper, we show that there exists a fixed-parameter algorithm for this problem, 
where an algorithm is called a {\em fixed-parameter algorithm} (or an {\em FPT algorithm}) if its running time is bounded by $f(k) (|V|+|E|)^{O(1)}$ for some function $f$. 
See~\cite{Cyganetal2015,FG2006,Nie2006} for more detail. 
Formally, our result is stated as follows. 

\begin{theorem}
\label{thm:fpt}
For a positive integer $t$, 
there exists a fixed-parameter algorithm for {\sc Parameterized} \ProblemName\ 
that runs in $(t+1)^{O(k^2+ t k)} |V| |E|$ time. 
%%that runs in $(k+t)^{O(kt)} (t+2)^{O(k^2)} |V| |E|$ time. 
In particular, the running time is $2^{O(k^2)} |V| |E|$ if $t$ is fixed. 
\end{theorem}

This result implies that there exists a fixed-parameter algorithm for the problem 
even when $t+k$ is the parameter. 
By using almost the same argument, 
%%we can show that {\sc Parameterized Minimum $(\alpha, \beta)$-Spanner Problem} 
we can show that a parameterized version of {\sc Minimum $(\alpha, \beta)$-Spanner Problem} 
is also fixed-parameter tractable. 
We define {\sc Parameterized Minimum $(\alpha, \beta)$-Spanner Problem} in the same way as 
{\sc Parameterized} \ProblemName. 

\begin{theorem}
\label{thm:fpt2}
For real numbers $\alpha \ge 1$ and $\beta \ge 0$, 
there exists a fixed-parameter algorithm for {\sc Parameterized Minimum $(\alpha, \beta)$-Spanner Problem} 
that runs in $(\alpha + \beta)^{O(k^2+ (\alpha + \beta) k)} |V| |E|$ time. 
%%that runs in $(k+\alpha+\beta)^{O(k(\alpha+\beta))} (\alpha+\beta)^{O(k^2)} |V| |E|$ time. 
\end{theorem}

\subsection{Related Work: Minimum Multiplicative Spanner Problem}
\label{sec:related}

As mentioned in Section~\ref{sec:spanner}, 
%%above, 
there are a lot of studies on {\sc Minimum Multiplicative $t$-Spanner Problem}, 
whereas only few results are known for \ProblemName. 

%%In this subsection, we describe known results on {\sc Minimum Multiplicative $t$-Spanner Problem}. 
%%%%{\sc Minimum Multiplicative $t$-Spanner Problem} is known to be 

%%This problem is 
{\sc Minimum Multiplicative $t$-Spanner Problem} is
NP-hard for any $t \ge 2$ in general graphs~\cite{CAI1994187,PS1989}, 
and there are several results on the problem for some graph classes. 
It is NP-hard even when the input graph is restricted to be planar~\cite{Brandes1997,kobayashi2018}. 
Cai and Keil~\cite{CK1994} showed that 
{\sc Minimum $2$-Spanner Problem} can be solved in linear time 
if the maximum degree of the input graph is at most $4$, 
whereas this problem is NP-hard even if the maximum degree is at most $9$. 
Venkatesan et al.~\cite{VENKATESAN1997143} revealed the complexity 
of {\sc Minimum Multiplicative $t$-Spanner Problem} for several graph classes
 such as chordal graphs, convex bipartite graphs, and split graphs. 
For the weighted version of the problem in which each edge has a positive integer length, 
Cai and Corneil~\cite{CaiCorneil1995} showed the NP-hardness of {\sc Minimum Multiplicative $t$-Spanner Problem} for $t > 1$. 

Another direction of research is to design approximation algorithms for {\sc Minimum Multiplicative $t$-Spanner Problem}.  
Kortsarz~\cite{Kort1994} gave an $O(\log n)$-approximation for $t=2$ 
and 
Elkin and Peleg~\cite{ElkinP2005} gave an $O(n^{2/(t+1)})$-approximation algorithm for $t > 2$.  
On the negative side, for any $t \ge 2$, it is shown in~\cite{Elkin2007} that
no $o(\log n)$-approximation algorithm exists unless $P=NP$. 
Dragan et al.~\cite{Dragan2011} gave an EPTAS for the problem in planar graphs. 
When the input graph is a $4$-connected planar triangulation, 
a PTAS is proposed for {\sc Minimum Multiplicative $2$-Spanner Problem} in \cite{DUCKWORTH200367}. 

A parameterized version of {\sc Minimum Multiplicative $t$-Spanner Problem} 
is introduced in~\cite{kobayashi2018}, and a fixed-parameter algorithm for it is presented in the same paper.

%%\subsection{Organization}

%%The remainder of this paper is organized as follows. 
%%In Section~\ref{sec:pre}, we give some preliminaries. 
%%In order to prove Theorem~\ref{thm:hardness}, 
%%we first show the NP-hardness of the problem of finding a minimum dominating set 
%%with additional constraints, which is described in Section~\ref{sec:dominate}. 
%%Then, by showing a reduction from \ProblemName\ to this problem, 
%%we give a proof of Theorem~\ref{thm:hardness} in Section~\ref{sec:spanner}. 
%%A crucial part of this reduction is Proposition~\ref{prop:dual}, 
%%which shows a relationship between a dominating set in a graph $G$
%%and a minimum $t$-spanner in the dual graph $G^*$. 
%%Note that a dominating set in $G$ corresponds to a set of faces in $G^*$, 
%%whereas a minimum $t$-spanner is a set of edges in $G^*$. 
%%Therefore, they look like completely unrelated objects. 
%%Although the proof of Proposition~\ref{prop:dual} is not so difficult, 
%%%%we believe that it is an important technical contribution 
%%it is not an easy task to find out 
%%this non-intuitive relationship between these two objects. 
%%In Section~\ref{sec:degree}, 
%%we observe properties of graphs used in Section~\ref{sec:spanner} and 
%%give proofs of Theorems~\ref{thm:hardness2} and~\ref{thm:hardness3}. 
%%Finally, in Section~\ref{sec:fpt}, we give a fixed-parameter algorithm for
%%{\sc Parameterized} \ProblemName\ and prove Theorem~\ref{thm:fpt}. 

\subsection{Organization}

The remainder of this paper is organized as follows. 
In Section~\ref{sec:pre}, we give some preliminaries. 
In Section~\ref{sec:fpt00}, 
we give an FPT algorithm for {\sc Parameterized} \ProblemName\ and prove Theorem~\ref{thm:fpt}. 
In Section~\ref{sec:extension}, 
we extend the argument in Section~\ref{sec:fpt00} to 
{\sc Parameterized Minimum $(\alpha, \beta)$-Spanner Problem}
and prove Theorem~\ref{thm:fpt2}. 
Finally, in Section~\ref{sec:conclusion}, we make a conclusion.

\section{Preliminaries}
\label{sec:pre}

%%We now give a formal definition of a $t$-spanner. 
In this paper, we deal with only undirected graphs with unit length edges. 
Since we can remove all the parallel edges and self-loops when we consider spanners, 
we assume that all the graphs in this paper are simple.  
%%Throughout the paper, 
Let $G=(V, E)$ be a graph. 
For $u, v \in V$, an edge connecting $u$ and $v$ is denoted by $uv$. 
For a subgraph $H$ of $G$, the set of vertices and the set of edges in $H$ are denoted by $V(H)$ and $E(H)$, respectively.  
For an edge $e \in E$, let $G-e$ denote the subgraph $G'=(V, E \setminus \{e\})$.  
We say that an edge set $F \subseteq E$ {\em contains} a path $P$ if $E(P) \subseteq F$. 
%%We say that two paths (or cycles) $P$ and $Q$ are distinct if $E(P) = E(Q)$. 
For a path $P$ and for two vertices $u, v \in V(P)$, 
let $P[u, v]$ denote the subpath of $P$ between $u$ and $v$. 
%%
%%In this paper, we sometimes identify a path or a cycle with its edge set. 
For $u, v \in V$, 
let $\dist_G (u,v)$ denote the distance of the shortest path 
between $u$ and $v$ in $G$. 
Note that the length of a path is the number of edges in it. 
If $G$ is clear from the context, $\dist_G(u, v)$ is simply denoted by $\dist(u, v)$. 
%%For a positive integer $t$, a subgraph $H = (V, E_H)$ of $G=(V, E)$ is said to be a {\em multiplicative $t$-spanner} 
%%if $\dist_H(u, v) \le t \cdot \dist_G(u, v)$ holds for any $u, v \in V$. 
For a positive integer $t$, a subgraph $H = (V, E_H)$ of $G=(V, E)$ is said to be an {\em additive $t$-spanner} 
if $\dist_H(u, v) \le \dist_G(u, v) + t$ or $\dist_G(u, v)= + \infty$ holds for any $u, v \in V$. 
For real numbers $\alpha \ge 1$ and $\beta \ge 0$, 
a subgraph $H = (V, E_H)$ of $G=(V, E)$ is said to be an {\em $(\alpha, \beta)$-spanner} 
if $\dist_H(u, v) \le \alpha \cdot \dist_G(u, v) + \beta$ or $\dist_G(u, v)= + \infty$ holds for any $u, v \in V$. 
In what follows, we may assume that the input graph $G=(V, E)$ is connected and $\dist_G(u, v)$ is finite for any $u, v \in V$, 
%%$|E| = \Omega(|V|)$,   
since we can deal with each connected component separately. 
For a positive integer $p$, let $[p] := \{1, \dots , p\}$.

\section{Proof of Theorem~\ref{thm:fpt} }
\label{sec:fpt00}

%%In this section, we give an FPT algorithm for {\sc Parameterized} \ProblemName\ and prove Theorem~\ref{thm:fpt}. 
%%We may assume that the input graph $G=(V, E)$ is connected and $|E| = \Omega(|V|)$,   
%%since we can deal with each connected component separately. 

\subsection{Outline}
\label{sec:outline}

In this subsection, 
we show an outline of our proof of Theorem~\ref{thm:fpt}. 
%%FPT algorithm for {\sc Parameterized} \ProblemName. 

Define $F \subseteq E$ as the set of all edges contained in cycles of length at most $t+2$. 
In other words, an edge $e=uv \in E$ is in $F$ if and only if $G-e$ contains a $u$-$v$ path of length at most $t+1$. 
By the definition, if $H=(V, E\setminus E')$ is an additive $t$-spanner of $G$, then 
${\rm dist}_{G-e}(u, v) \le {\rm dist}_H(u, v) \le {\rm dist}_G(u, v) + t = t+1$ holds for each $e=uv \in E'$, 
which implies that $E' \subseteq F$. 
Thus, if $|F|$ is small, then we can solve {\sc Parameterized} \ProblemName\ 
by checking whether $H = (V, E \setminus E')$ is an additive $t$-spanner of $G$ or not
for every subset $E'$ of $F$ with $|E'| = k$.   

If $|F|$ is sufficiently large, then 
there exist many cycles of length at most $t+2$. 
In what follows, we show that 
if $G$ has many cycles of length at most $t+2$, then 
there always exists $E' \subseteq E$ with $|E'| = k$ 
such that $H = (V, E \setminus E')$ is an additive $t$-spanner of $G$. 
To this end, we prove the following statements in Sections~\ref{sec:fpt1}--\ref{sec:fpt3}, respectively. 
\begin{itemize}
\item
If there are many cycles of length at most $t+2$, then 
we can find either 
many edge-disjoint cycles of length at most $t+2$ 
or 
a desired set $E' \subseteq E$ (Section~\ref{sec:fpt1}).  
\item
If there are many edge-disjoint cycles of length at most $t+2$, then 
we can construct a sequence of edge-disjoint cycles with a certain condition (Section~\ref{sec:fpt2}). 
\item
If we have a sequence of edge-disjoint cycles with a certain condition, then 
we can construct a desired set $E' \subseteq E$ (Section~\ref{sec:fpt3}). 
\end{itemize}
Finally, in Section~\ref{sec:fpt4}, we put them together and describe our entire algorithm. 

\subsection{Finding Edge-disjoint Cycles}
\label{sec:fpt1}

The objective of this subsection is 
to show that if there are many cycles of length at most $t+2$, then 
we can find either 
many edge-disjoint cycles of length at most $t+2$ 
or 
a desired set $E' \subseteq E$.  
We first show the following lemma. 
%%, whose proof is postponed to the appendix due to the space constraint. 

\begin{lemma}
\label{lem:11}
For positive integers $k$ and $\ell$, there exists an integer $f_1(k, \ell) = (k \ell)^{O(\ell)}$ satisfying the following condition. 
For any pair of distinct vertices $u, v \in V$ in a graph $G=(V, E)$, 
if there exists a set $\mathcal P$ of $u$-$v$ paths of length at most $\ell$ with $|\mathcal{P}| \ge f_1(k, \ell)$,  
then $G$ contains two distinct vertices $u', v'\in V$ and $k$ edge-disjoint $u'$-$v'$ paths of length 
at most $\ell - \dist(u, u')-\dist(v, v')$. 
\end{lemma}

\begin{proof}
We show that $f_1(k, \ell) : = 2 (k \ell^3)^{\ell-1}$ satisfies the condition by induction on $\ell$. 
The claim is obvious when $\ell=1$, because $|\mathcal{P}| \le 1$ holds as $G$ is simple and $f_1(k, 1)=2$. 
Thus, it suffices to consider the case of $\ell \ge 2$. 
Let $\mathcal P$ be a set of $u$-$v$ paths of length at most $\ell$ with $|\mathcal{P}| \ge f_1(k, \ell)$. 
We consider the following two cases separately. 

We first consider the case when $|\{ P \in \mathcal P \mid e \in E(P) \}| < \frac{f_1(k, \ell)}{k \ell}$ holds for any $e \in E$.  
In this case, $|\{ Q \in \mathcal P \mid E(P) \cap E(Q) \not= \emptyset \}| < \frac{f_1(k, \ell)}{k}$ for any $P \in \mathcal P$. 
This shows that we can take $k$ edge-disjoint $u$-$v$ paths in $\mathcal P$ by a greedy algorithm
(i.e., repeatedly taking a $u$-$v$ path $P$ in $\mathcal P$ and removing all the paths sharing an edge with $P$).  
They form a desired set of paths in which $u'=u$ and $v'=v$. 

We next consider the case when there exists an edge $e = xy \in E$ such that 
$|\{ P \in \mathcal P \mid e \in E(P) \}| \ge \frac{f_1(k, \ell)}{k \ell} = 2 \ell^2 (k \ell^3)^{\ell-2}$.  
Without loss of generality, we may assume that $x \not\in \{u, v\}$. 
For $i=1, \dots , \ell-1$, let $\mathcal P^i_{ux}$ be the set of all $u$-$x$ paths of length $i$ 
and $\mathcal P^i_{xv}$ be the set of all $x$-$v$ paths of length $i$. 
Then, since each path $P \in \mathcal P$ containing $e$ can be divided into a $u$-$x$ path and an $x$-$v$ path, 
we obtain 
$$
\sum_{i+j \le \ell} |\mathcal P^i_{ux}| \cdot |\mathcal P^j_{xv}| \ge |\{ P \in \mathcal P \mid e \in E(P) \}| \ge 2 \ell^2 (k \ell^3)^{\ell-2}. 
$$
Since the number of pairs $(i, j)$ with $i+j \le \ell$ is at most $\frac{\ell (\ell -1)}{2} < \frac{\ell^2}{2}$, 
there exist $i, j \in [\ell - 1]$ with $i+j \le \ell$ such that 
$$
|\mathcal P^i_{ux}| \cdot |\mathcal P^j_{xv}| \ge 2 \ell^2 (k \ell^3)^{\ell-2} \cdot \frac{2}{\ell^2} \ge 2 (k \ell^3)^{i-1} \cdot 2 (k \ell^3)^{j-1} \ge f_1(k, i) \cdot f_1(k, j).  
$$
This shows that either $|\mathcal P^i_{ux}| \ge f_1(k, i)$ or $|\mathcal P^j_{xv}| \ge f_1(k, j)$ holds. 
By induction hypothesis,
if $|\mathcal P^i_{ux}| \ge f_1(k, i)$, then 
there exist $u', v'\in V$ and $k$ edge-disjoint $u'$-$v'$ paths of length 
at most 
\begin{align*}
i - \dist(u, u')-\dist(x, v') 
&\le \ell - j - \dist(u, u')-\dist(x, v') \\
&\le \ell - \dist(x, v) - \dist(u, u')-\dist(x, v') \\
&\le \ell - \dist(u, u')-\dist(v, v').  
\end{align*}
Thus, they form a desired set of paths. 
The same argument can be applied when $|\mathcal P^j_{xv}| \ge f_1(k, j)$.
\end{proof}

By using this lemma, we obtain the following proposition. 

\begin{proposition}
\label{prop:12}
Let $G=(V, E)$ be a graph and 
$\mathcal C$ be 
a set of cycles of length at most $t+2$. 
Let $N$ be a positive integer and $f_1$ be a function as in Lemma~\ref{lem:11}. 
If $|\mathcal{C}| \ge N (t+2) f_1(k+t+1, t+1)$, then 
we have one of the following. 
\begin{itemize}
\item
There exist 
$N$ edge-disjoint cycles 
in $\mathcal C$. 
\item
There exists 
$E' \subseteq \bigcup_{C \in \mathcal{C}} E(C)$ with $|E'| = k$ such that 
$H=(V, E\setminus E')$ is an additive $t$-spanner of $G$. 
\end{itemize}
\end{proposition}

\begin{proof}
For each edge $e \in E$, let $\mathcal{C}_e := \{ C \in \mathcal C \mid e \in E(C) \}$. 
We first consider the case when $|\mathcal{C}_e| < f_1(k+t+1, t+1)$ holds for any $e \in E$.  
In this case, $|\{ C' \in \mathcal C \mid E(C) \cap E(C') \not= \emptyset \}| < (t+2) f_1(k+t+1, t+1)$ for any $C \in \mathcal{C}$. 
This shows that we can take $N$ edge-disjoint cycles in $\mathcal C$ by a greedy algorithm 
(i.e., repeatedly taking a cycle $C$ in $\mathcal C$ and removing all the cycles sharing an edge with $C$), 
because $|\mathcal C| \ge N (t+2) f_1(k+t+1, t+1)$. 
%%That is, we obtain $N$ edge-disjoint cycles of length at most $t+2$.

We next consider the case when there exists an edge $e=uv \in E$ such that 
$|\mathcal{C}_e| \ge f_1(k+t+1, t+1)$.  
Since $\mathcal P := \{C-e \mid C \in \mathcal{C}_e \}$ consists of 
$u$-$v$ paths of length at most $t+1$, by Lemma~\ref{lem:11}, 
$G$ contains two vertices $u', v'\in V$ and a set $\mathcal{P}'$ of $k+t+1$ edge-disjoint $u'$-$v'$ paths of length 
at most $t' := t+1 - \dist_G(u, u')-\dist_G(v, v')$. 
Let $Q_u$ and $Q_v$ be a shortest $u$-$u'$ path and a shortest $v$-$v'$ path, respectively. 
Since $|E(Q_u)| + |E(Q_v)| + 1 = t+2-t' \le t+1$, 
there exists $\mathcal{P}'' \subseteq \mathcal{P}'$ with $|\mathcal{P}''|=k$ 
such that each path in $\mathcal{P}''$ does not contain edges in $E(Q_u) \cup E(Q_v) \cup\{e\}$. 
Let $P_1, \dots , P_k$ denote the paths in $\mathcal{P}''$. 
For $i=1, \dots , k$, let $e_i$ be the middle edge of $P_i$ (see Fig.~\ref{fig:01}). 
Formally, we take $e_i = x_i y_i$ so that 
$P_i [u', x_i]$ contains $\lfloor \frac{|E(P_i)|-1}{2} \rfloor \le \lfloor \frac{t'-1}{2}\rfloor$ edges and
$P_i [y_i, v']$ contains $\lceil \frac{|E(P_i)|-1}{2} \rceil \le \lceil \frac{t'-1}{2}\rceil$ edges. 
Define $E' := \{e_1, \dots , e_k\}$ and consider the graph $H = (V, E \setminus E')$. 

\begin{figure}[htbp]
  \begin{center}
	\vspace{-10pt}
    \includegraphics[width=5.0cm]{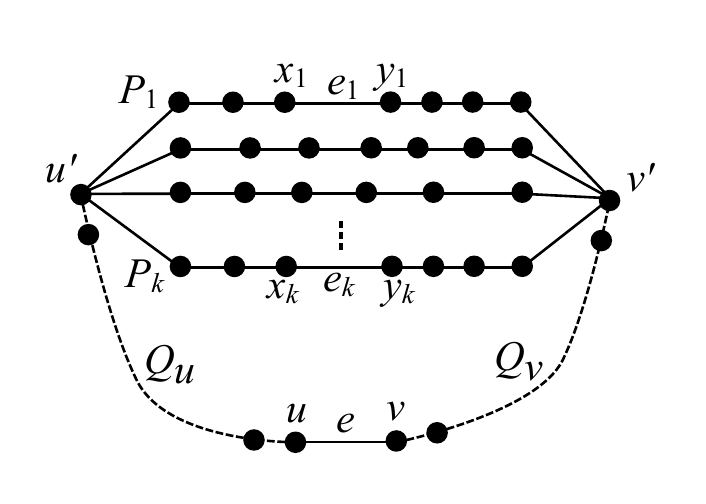}
	\vspace{-10pt}
    \caption{Definition of $e_1, \dots , e_k$ in Proposition~\ref{prop:12}.}
    \label{fig:01}
  \end{center}
	\vspace{-10pt}
\end{figure}

We now show that $H$ is an additive $t$-spanner of $G$. 
Let $x$ and $y$ be distinct vertices in $V$ and let $P$ be a shortest $x$-$y$ path in $G$. 
If $E(P) \cap E' = \emptyset$, then it is obvious that $\dist_{H}(x, y) = \dist_{G}(x, y)$. 
If $E(P) \cap E' \not= \emptyset$, then let $P[z, z']$ be the unique minimal subpath of $P$ 
that contains all edges in $E(P) \cap E'$, where 
$x, z, z'$, and $y$ appear in this order along $P$.
%%Since $z$ and $z'$ are the endpoints of edges in $E'$, we have $\dist_{H} (z, z') \le t+1$ 
Since $z, z' \in \{x_1, y_1, \dots , x_k, y_k\}$, we have $\dist_{H} (z, z') \le t+1$ 
by observing that 
\begin{itemize}
\item
$\dist_{H} (x_i, x_j) \le 2 \cdot \lfloor \frac{t'-1}{2}\rfloor  \le t+1$ for any $i, j \in [k]$,  
\item
$\dist_{H} (y_i, y_j) \le 2 \cdot \lceil \frac{t'-1}{2}\rceil  \le t+1$ for any $i, j \in [k]$, and 
\item
%%Since $E(Q_u) \cup E(Q_v) \cup\{e\}$ contains a $u'$-$v'$ path $Q$ of length at most $t+2-t'$, 
%%we have that 
%%$\dist_{H} (x_i, y_j) \le |E(Q)| + \lceil \frac{t'-1}{2}\rfloor + \lceil \frac{t'-1}{2}\rceil \le t+1$ for any $i, j \in [k]$. 
$\dist_{H} (x_i, y_j) \le  \dist_{H} (x_i, u') +  \dist_{H} (u', v') + \dist_{H} (v', y_i) \le \lfloor \frac{t'-1}{2}\rfloor + (t+2-t') + \lceil \frac{t'-1}{2}\rceil \le t+1$ for any $i, j \in [k]$. 
\end{itemize}
Therefore,  
\begin{align*}
\dist_{H} (x, y)
& \le \dist_{H} (x, z) + \dist_{H} (z, z') + \dist_{H} (z', y)  \\ 
&\le \dist_{G} (x, z) + t+1 + \dist_{G} (z', y) \\
& = \dist_{G} (x, y)  - \dist_{G} (z, z') + t+1 \\ 
&\le \dist_{G} (x, y) + t, 
\end{align*}
which shows that $H$ is an additive $t$-spanner of $G$. 
%%
%%Since $E(Q_u) \cup E(Q_v) \cup\{e\}$ contains a $u'$-$v'$ path $Q$ of length at most $t+2-t'$, 
%%we have that 
%%
%%\begin{enumerate}
%%\item[(i)]
%%$\dist_{H} (x_i, y_i) \le |E(Q)| + |E(P_i)| - 1 \le t+1$ for each $i$, and
%%\item[(ii)]
%%$\dist_{H} (z_i, z_j) \le |E(Q)| + \frac{t'}{2} + \frac{t'}{2} \le t+2$ for any $z_i \in \{x_i, y_i\}$ and any $z_j \in \{x_j, y_j\}$ with $i \not= j$. 
%%\end{enumerate}
%%
%%Let $x$ and $y$ be distinct vertices in $V$ and let $P$ be a shortest $x$-$y$ path in $G$. 
%%If $E(P) \cap E' = \emptyset$, then it is obvious that $\dist_{H}(x, y) = \dist_{G}(x, y)$. 
%%If $|E(P) \cap E'| = 1$, then $\dist_{H}(x, y) \le \dist_{G}(x, y) + t + 1$ holds by (i). 
%%If $|E(P) \cap E'| \ge 2$, then let $P[z, z']$ be the unique minimal subpath of $P$ 
%%that contains all edges in $E(P) \cap E'$. 
%%Since $z$ and $z'$ are the endpoints of edges in $E'$, we have $\dist_{H} (z, z') \le t+2$ by (ii). 
%%This shows that $H$ contains an $x$-$y$ path of length at most $|E(P)| - |E(P[z, z'])| + t + 2 \le \dist_G(x, y) + t$, 
%%where we used $|E(P[z, z'])| \ge |E(P) \cap E'| \ge 2$. 
%%Therefore, $H$ is an additive $t$-spanner of $G$. 
\end{proof}

%%We note that 
%%the proofs of Lemma~\ref{lem:11} and Proposition~\ref{prop:12} are constructive, and hence
%%we can obtain either $N$ edge-disjoint cycles or an additive $t$-spanner in 
%%$O((t+2) |\mathcal{C}| |E|)$ time. 
%%Although $|\mathcal{C}|$ might be huge in general, 
%%we can apply the same argument as Proposition~\ref{prop:12}
%%to a subset $\mathcal{C}'$ of $\mathcal{C}$ with $|\mathcal{C'}| = N (t+2) f_1(k+t+1, t+1)$. 
%%Thus, 
%%we can obtain either $N$ edge-disjoint cycles or an additive $t$-spanner in 
%%$O((t+2) |\mathcal{C}'| |E|) = O( N (t+2)^2 f_1(k+t+1, t+1) |E|)$ time. 

\subsection{Finding a Good Sequence of Cycles }
\label{sec:fpt2}

In this subsection, we construct 
a sequence of edge-disjoint cycles with a certain condition
when we are given many edge-disjoint cycles.

Let $\mathcal C$ be a set of edge-disjoint cycles of length at most $t+2$. 
For a vertex $v \in V$ and a cycle $C \in \mathcal C$, 
let $P(v, C)$ be a shortest path from $v$ to $V(C)$. 
By choosing an appropriate shortest path for each $v \in V$, 
we may assume that $\bigcup_{v \in V} E(P(v, C))$ forms a forest 
for any cycle $C \in \mathcal C$. 
The objective of this subsection is to find 
a sequence $(C_1, \dots , C_p)$ of 
distinct $p$ cycles $C_1, \dots, C_p \in \mathcal C$ satisfying the following condition: 

\begin{description}
\item[($\star$)]
For any $h, i, j \in [p]$ with $h < i < j$ and for any vertex $v \in V(C_h)$, 
it holds that $E(P(v, C_i)) \cap E(C_j) \not= \emptyset$. 
\end{description}

Roughly speaking, this condition means that 
if $h < i < j$, then removing edges in $E(C_j)$ does not affect the 
distance between $C_h$ and $C_i$.

\begin{lemma}
\label{lem:13}
For any positive integers $t$ and $p$, there exists an integer $f_2(t, p)=O(t^2 p^4)$ satisfying the following condition. 
If $\mathcal C$ is a set of $f_2(t, p)$ edge-disjoint cycles of length at most $t+2$, then 
there exists a sequence $(C_1, \dots , C_p)$ of 
distinct $p$ cycles $C_1, \dots, C_p \in \mathcal C$
satisfying the condition {\rm ($\star$)}. 
\end{lemma}

\begin{proof}
We show that $f_2(t, p) := 27 (t+2) (3t+1) p^4$ satisfies the condition in the lemma. 
Let $\mathcal C$ be a set of $f_2(t, p)$ edge-disjoint cycles of length at most $t+2$. 
We consider the following two cases separately. 

\medskip

\textbf{Case 1.}
Suppose that there exist a vertex $v \in V$ and a cycle $C^* \in \mathcal{C}$
such that 
$| E(P(v, C^*)) \cap \bigcup_{C \in \mathcal C} E(C)| \ge (3t+1) p$. 
In this case, we can take edges $e_1, \dots , e_p$ in $E(P(v, C^*)) \cap \bigcup_{C \in \mathcal C} E(C)$
such that $e_1=x_1 y_1, e_2=x_2 y_2, \dots , e_p=x_p y_p$ appear in this order along $P(v, C^*)$ and
the subpath of $P(v, C^*)$ between $x_i$ and $x_{i+1}$ contains at least $3t+1$ edges for $i=1, \dots , p-1$ (see Fig.~\ref{fig:02}). 
For $i=1, \dots , p$, let $C_i \in \mathcal C$ be the cycle containing $e_i$. 
Note that $C_i$ and $C_j$ are distinct if $i \not= j$, since $\dist_{G} (x_i, x_j) \ge 3 t +1 > |E(C_i)|$.

\begin{figure}%%[htbp]
 \begin{minipage}{0.5\hsize}
  \begin{center}
%%\vspace{-7pt}
    \includegraphics[width=6.5cm]{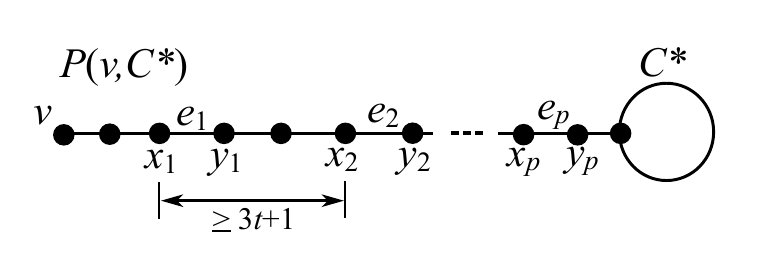}
%%\vspace{-3pt}
    \caption{Definition of $e_1, \dots , e_{p}$.}
    \label{fig:02}
  \end{center}
 \end{minipage}
 \begin{minipage}{0.5\hsize}
  \begin{center}
%%\vspace{-7pt}
    \includegraphics[width=6.7cm]{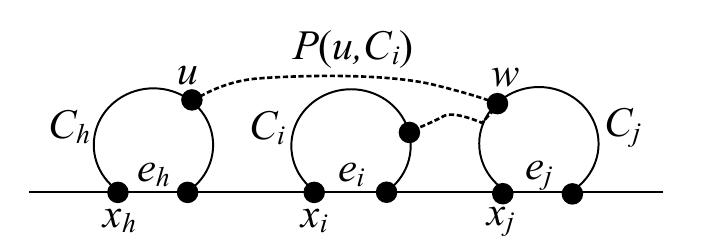}
%%\vspace{-3pt}
    \caption{Definition of $w$.}
    \label{fig:03}
  \end{center}
 \end{minipage}
\end{figure}

%%\begin{figure}[htbp]
%%  \begin{center}
%%    \includegraphics[width=6.5cm]{fig02.pdf}
%%    \caption{Definition of $e_1, \dots , e_{p}$ in Case 1 of Lemma~\ref{lem:13}.}
%%    \label{fig:02}
%%  \end{center}
%%\end{figure}

We now show that $(C_1, \dots, C_p)$ satisfies the condition ($\star$).
Assume to the contrary that there exist 
indices $h, i, j \in [p]$ with $h< i < j$ and 
a vertex $u \in V(C_h)$ such that
$E(P(u, C_i)) \cap E(C_j) \not= \emptyset$. 
Let $w$ be the first vertex in $V(C_j)$ when we traverse $P(u, C_i)$ from $u$ to $V(C_i)$ (see Fig.~\ref{fig:03}). 
Then, we have 
\begin{align*}
\dist (x_h, x_i) + t &\ge \dist (x_h, x_i) + \dist (u, x_h) \ge \dist (u, x_i) \\ 
                        &\ge |E(P(u, C_i))| \ge  \dist (u, w) \\
                        &\ge \dist (x_h, x_j) -  \dist (x_h, u)  - \dist (w, x_j) \ge \dist (x_h, x_j) - 2t \\ 
                        &\ge (\dist (x_h, x_i) + 3t+1) - 2t = \dist (x_h, x_i) + t + 1
\end{align*}
by using $\dist (u, x_h) \le \lfloor \frac{|E(C_h)|}{2} \rfloor \le t$ and $\dist (x_j, w) \le \lfloor \frac{|E(C_j)|}{2} \rfloor \le t$, 
which is a contradiction. 
Therefore, $(C_1, \dots, C_p)$ satisfies the condition ($\star$).

%%\begin{figure}[htbp]
%%  \begin{center}
%%    \includegraphics[width=6.5cm]{fig03.pdf}
%%    \caption{Definition of $w$.}
%%    \label{fig:03}
%%  \end{center}
%%\end{figure}

\medskip

\textbf{Case 2.}
Suppose that $| E(P(v, C^*)) \cap \bigcup_{C \in \mathcal C} E(C)| < (3t+1) p$ holds 
for any vertex $v \in V$ and for any cycle $C^* \in \mathcal{C}$,  
which implies that $| \{ C \in \mathcal{C} \mid E(P(v, C^*)) \cap E(C) \not= \emptyset \} | < (3t+1) p$. 
%%
%%Then, for any $C_{h}, C_{i} \in \mathcal{C}$, it holds that 
%%\begin{equation}
%%| \{ C \in \mathcal{C} \mid E(P(v, C_{i})) \cap E(C) \not= \emptyset \mbox{ for some } v \in V(C_{h})  \} | < 
%%(t+2) (3t+1) p. \label{eq:01}
%%\end{equation}
%%
We define $\mathcal F_3 \subseteq \mathcal{C}^3$ by 
$$
\mathcal F_3 := 
\{ (C_h, C_i, C_j) \mid  C_h, C_i, C_j \in \mathcal{C}, 
\ E(P(v, C_i)) \cap E(C_j) \not= \emptyset \mbox{ for some } v \in V(C_h) \}. 
$$
%%Then, $|\mathcal F_3| <  (t+2) (3t+1) p | \mathcal{C} |^2$ by (\ref{eq:01}). 
Then, it holds that 
\begin{align*}
|\mathcal F_3| 
&=  \sum_{C_h \in \mathcal C} \sum_{C_i \in \mathcal C} | \{ C_j \in \mathcal{C} \mid E(P(v, C_i)) \cap E(C_j) \not= \emptyset \mbox{ for some } v \in V(C_h) \}| \\
&\le  \sum_{C_h \in \mathcal C} \sum_{C_i \in \mathcal C} \sum_{v \in V(C_h)} | \{ C_j \in \mathcal{C} \mid E(P(v, C_i)) \cap E(C_j) \not= \emptyset \}| \\
&<  \sum_{C_h \in \mathcal C} \sum_{C_i \in \mathcal C} \sum_{v \in V(C_h)} (3t+1) p  \\
&\le (t+2) (3t+1) p | \mathcal{C} |^2.
\end{align*}
We note that $(C_1, \dots , C_p)$ satisfies the condition ($\star$) if and only if 
$(C_{h}, C_{i}, C_{j}) \not \in \mathcal F_3$ holds 
for any $h, i, j \in [p]$ with $h < i < j$. 
That is, $\mathcal F_3$ represents the set of forbidden orderings of three cycles. 
We define $\mathcal F_2 \subseteq \mathcal{C}^2$ and $\mathcal F_1 \subseteq \mathcal{C}$ by 
\begin{align*}
\mathcal F_2 &:= 
\bigg\{ (C_{h}, C_{i}) \in \mathcal{C}^2 \ \bigg| \ |\{C \in \mathcal{C} \mid (C_{h}, C_{i}, C) \in \mathcal{F}_3 \} | \ge \frac{|\mathcal{C}| }{3 p^2}  \bigg\}, \\ 
\mathcal F_1 &:= 
\bigg \{ C_{h} \in \mathcal{C} \ \bigg | \  |\{C \in \mathcal{C} \mid (C_{h}, C) \in \mathcal{F}_2 \} | \ge \frac{|\mathcal{C}|}{3p}  \bigg\}. 
\end{align*}
Then, we have
\begin{align*}
|\mathcal{F}_2| &\le |\mathcal{F}_3| \cdot \frac{3 p^2}{|\mathcal{C}|} < 3(t+2) (3t+1) p^3 | \mathcal{C}|, \\ 
|\mathcal{F}_1| &\le |\mathcal{F}_2| \cdot \frac{3 p}{|\mathcal{C}|} < 9 (t+2) (3t+1) p^4 \le \frac{|\mathcal{C}|}{3}. 
\end{align*}

In order to obtain $(C_1, \dots , C_p)$ satisfying the condition ($\star$), 
we construct a sequence of cycles satisfying additional conditions. 
\begin{claim}
For each $q \in [p]$, there exists a sequence $(C_1, \dots , C_q)$  
of $q$ distinct cycles $C_1, \dots , C_q \in \mathcal C$ satisfying the following conditions: 
\begin{itemize}
\item
$C_{h} \not\in \mathcal{F}_1$ for any $h \in [q]$, 
\item
$(C_{h}, C_{i}) \not\in \mathcal{F}_2$ for any $h, i \in [q]$ with $h < i$, and
\item
$(C_{h}, C_{i}, C_{j}) \not\in \mathcal{F}_3$ for any $h, i, j \in [q]$ with $h < i < j$.  
\end{itemize}
\end{claim}

\begin{proof}[Proof of the claim]
We show the claim by induction on $q$. 
When $q=1$, we can choose $C_1 \in \mathcal{C} \setminus \mathcal{F}_1$ arbitrarily. 
Suppose that we have $C_1, \dots , C_q \in \mathcal{C}$ satisfying the conditions in the claim, where $q \le p-1$.  
Then, we have that
\begin{align*}
& N_2 := |\{ C \in \mathcal{C} \mid (C_{h}, C) \in \mathcal{F}_2 \mbox{ for some } h \in [q] \}|  \le q \cdot  \frac{|\mathcal{C}| }{3 p}  < \frac{|\mathcal{C}|}{3} - p, \\
& N_3 := |\{ C \in \mathcal{C} \mid (C_{h}, C_{i}, C) \in \mathcal{F}_3 \mbox{ for some $h, i \in [q]$ with }  h < i  \}|  \le  q^2  \cdot  \frac{|\mathcal{C}|}{3 p^2}  
< \frac{|\mathcal{C}|}{3} 
\end{align*}
by the definitions of $\mathcal F_1$ and $\mathcal{F}_2$. 
Since $|\mathcal{C}| - |\mathcal{F}_1| - N_2 - N _3 > p \ge q+1$,
there exists a cycle $C_{q+1} \in \mathcal C$ that is different from $C_1, \dots , C_q$ such that 
$(C_1, \dots , C_q, C_{q+1})$ satisfies the conditions in the claim. 
This shows the claim by induction on $q$. 
\end{proof}

By this claim, there exists a sequence $(C_1, \dots , C_p)$  
of $p$ distinct cycles $C_1, \dots , C_p \in \mathcal C$ such that
$(C_{h}, C_{i}, C_{j}) \not\in \mathcal{F}_3$ for any $h, i, j \in [p]$ with $h < i < j$, 
which means that $(C_1, \dots , C_p)$ satisfies the condition ($\star$).
\end{proof}

\subsection{Constructing an Additive $t$-Spanner}
\label{sec:fpt3}

In this subsection, we show that we can construct an additive $t$-spanner of $G$
by using a sequence of edge-disjoint cycles satisfying the condition ($\star$).

\begin{lemma}
\label{lem:14}
For any positive integers $t$ and $k$, there exists an integer $f_3(t, k)=(t+2)^{O(k)}$ satisfying the following condition. 
If there exists a sequence $(C_1, \dots, C_p)$ of $p=f_3(t, k)$ edge-disjoint cycles of length at most $t+2$ satisfying the condition {\rm ($\star$)}, then 
there exists an edge set $E' \subseteq \bigcup_{i \in [p]} E(C_i)$ with $|E'| = k$ such that 
$H=(V, E \setminus E')$ is an additive $t$-spanner of $G$. 
\end{lemma}

\begin{proof}
We show that $p= f_3(t, k):=k (t+2)^{k-1}$ satisfies the condition. 
For each edge $e \in E$, define 
$$
I(e) := \{ i \in [p] \mid e \not\in \bigcup_{v \in V} E(P(v, C_i)) \}. 
$$
Since we assumed that $\bigcup_{v \in V} E(P(v, C_i))$ forms a forest for each $i \in [p]$, 
for any cycle $C$, there exists an edge $e \in E(C)$ such that $i \in I(e)$. 
In other words, $\bigcup_{e \in E(C)} I(e) = [p]$ for any cycle $C$. 
We prove the lemma by 
showing that Algorithm~\ref{alg1}
always finds an edge set $E' \subseteq \bigcup_{i \in [p]} E(C_i)$ with $|E'| = k$ such that 
$H=(V, E \setminus E')$ is an additive $t$-spanner of $G$. 

	\begin{algorithm}[h]
	\SetKwInOut{Input}{Input}\SetKwInOut{Output}{Output}
	\Input{A sequence $(C_1, \dots, C_p)$ of edge-disjoint cycles of length at most $t+2$ with the condition ($\star$)}
	\Output{An edge set $E'\subseteq \bigcup_{i \in [p]} E(C_i)$ with $|E'| = k$ such that 
$H=(V, E \setminus E')$ is an additive $t$-spanner}
	$I_0 := [p]$ \\
	\For{$i=1, \dots , k$}{
		Let ${\rm ind}(i)$ be the minimum index in $I_{i-1}$ \\ 
         $C'_i := C_{{\rm ind}(i)}$ \\
		Choose an edge $e_i \in E(C'_{i})$ that maximizes $|(I_{i-1} \setminus \{{\rm ind}(i)\}) \cap I(e_i)|$ \\
 		$I_i := (I_{i-1} \setminus \{{\rm ind}(i)\}) \cap I(e_i)$ 
	}
	Return $E' := \{e_1, \dots , e_k\}$ 
	\caption{Constructing an additive $t$-spanner from a sequence with ($\star$)}
	\label{alg1}
	\end{algorithm} 

We first show that the algorithm returns a set of $k$ edges. 
For $i=1, \dots , k-1$, 
since $\bigcup_{e \in E(C'_{i})} I(e) = [p]$ and $|E(C'_{i})| \le t+2$, 
we have that $|I_i| \ge \frac{|I_{i-1}|-1}{t+2}$.  
By combining this with $|I_0| = k (t+2)^{k-1}$, 
we see that $|I_i| \ge (k-i) (t+2)^{k-i-1}$ for each $i$ by induction. 
In particular, $|I_{k-1}| \ge 1$ holds, and hence the algorithm 
returns a set $E' = \{e_1, \dots , e_k\}$.

We next show that $H=(V, E \setminus E')$ is an additive $t$-spanner. 
Let $x$ and $y$ be distinct vertices in $V$ and let $P$ be a shortest $x$-$y$ path in $G$. 
If $E(P) \cap E' = \emptyset$, then it is obvious that $\dist_{H}(x, y) = \dist_{G}(x, y)$. 
If $E(P) \cap E' = \{e_i\}$ for some $i \in \{1, \dots , k\}$, then  
$\dist_{H}(x, y) \le \dist_{G}(x, y) + t + 1$ holds, 
because 
$(E(P) \setminus \{e_i\}) \cup (E(C'_{i}) \setminus \{e_i\})$
contains an $x$-$y$ path.  

Thus, it suffices to consider the case when $|E(P) \cap E'| \ge 2$. 
Let $P[z, z']$ be the unique minimal subpath of $P$ 
that contains all edges in $E(P) \cap E'$, where 
$x, z, z'$, and $y$ appear in this order along $P$.  
Then, $z$ and $z'$ are the endpoints of edges $e_h$ and $e_i$ in $E(P) \cap E'$, respectively.  
We may assume that $h < i$ by changing the roles of $x$ and $y$ if necessarily. 
We now observe the following properties of $P(z, C'_{i})$. 
\begin{itemize}
\item
Since $(C_1, \dots , C_p)$ satisfies ($\star$), 
$(C'_1, \dots , C'_k)$ also satisfies ($\star$). 
This shows that 
$P(z, C'_{i})$ does not contain edges in $E(C'_{j})$ for any $j > i$, because $z \in V(C'_h)$ and $h < i$. 
In particular, $P(z, C'_{i})$ does not contain $e_{j}$ for any $j > i$.  
\item
Since ${\rm ind}(i) \in I_{i-1} \subseteq I(e_1) \cap I(e_2) \cap \dots \cap I(e_{i-1})$ by the algorithm, 
$P(z, C'_{i})$ does not contain $e_j$ for any $j < i$. 
\item
It is obvious that $P(z, C'_{i})$ does not contain $e_i$ by the definition of $P(z, C'_{i})$. 
\end{itemize}
Hence, $P(z, C'_{i})$ does not contain edges in $E'$, which means that 
$P(z, C'_{i})$ is a path in $H$ (see Fig.~\ref{fig:04}). 
Since $C'_{i} - e_i$ contains a path connecting 
an endpoint of $P(z, C'_{i})$ and $z'$, we have that 
\begin{align*}
\dist_{H}(x, y) &\le \dist_{H}(x, z) + \dist_{H}(z, z') + \dist_{H}(z', y) \\ 
                    &\le \dist_{G}(x, z) + (|E(P(z, C'_{i}))| + |E(C'_{i})| -1) + \dist_{G}(z', y) \\ 
                    &\le \dist_{G}(x, z) + |E(P[z, z']) \setminus \{e_i\}| + t+1 + \dist_{G}(z', y) \\ 
                    &\le \dist_{G}(x, z) + (\dist_{G}(z, z') -1) + t+1 + \dist_{G}(z', y) \\
				  &= \dist_{G}(x, y)+t. 
\end{align*}
%%where we used $|E(P(z, C'_{i}))| \le |E(P[z, z']) \setminus \{e_i\}| = \dist$ in the inequality. 
Therefore, $H$ is an additive $t$-spanner of $G$. 
\end{proof}

\begin{figure}[htbp]
  \begin{center}
    \includegraphics[width=6.3cm]{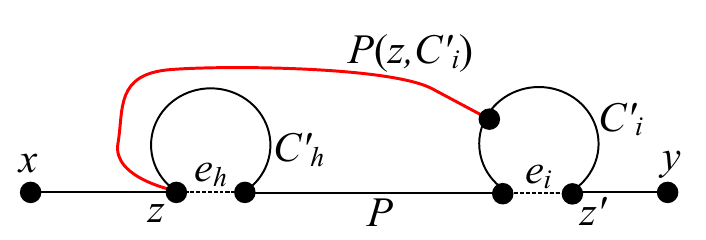}
    \caption{Proof of Lemma~\ref{lem:14}.}
    \label{fig:04}
  \end{center}
\end{figure}

%%\subsection{Proof of Theorem~\ref{thm:fpt}}
\subsection{The Entire Algorithm}
\label{sec:fpt4}

In this subsection, we describe our entire algorithm for {\sc Parameterized} \ProblemName\ and prove Theorem~\ref{thm:fpt}
by using Proposition~\ref{prop:12} and Lemmas~\ref{lem:13} and~\ref{lem:14}. 
Define 
\begin{align*}
p &: = f_3(t, k),  &
N &:= f_2(t, p), &
f_4(t, k) &:= N (t+2)^2 f_1(k+t+1, t+1), 
\end{align*}
where $f_1$, $f_2$, and $f_3$ are as in Lemmas~\ref{lem:11},~\ref{lem:13}, and~\ref{lem:14}, respectively.
Then, $N = (t+2)^{O(k)}$ and $f_1(k+t+1, t+1) = (kt)^{O(t)}$, and hence 
$f_4(t, k) = (t+2)^{O(k)} \cdot (kt)^{O(t)}  =  (t+2)^{O(k)} \cdot k^{O(t)} \cdot t^{O(t)}= (t+1)^{O(k+t)}$.
Note that we can simply denote $f_4(t, k) = t^{O(k+t)}$ unless $t=1$. 

Our algorithm for {\sc Parameterized} \ProblemName\ is as follows. 
We first compute the set $F \subseteq E$ of all edges contained in cycles of length at most $t+2$. 
Note that we can do it in $O(|V| |E|)$ time 
by applying the breadth first search from each vertex. 

As described in Section~\ref{sec:outline}, if $H=(V, E\setminus E')$ is an additive $t$-spanner of $G$ for $E' \subseteq E$, then $E' \subseteq F$ holds. 
Thus, if $|F| \le f_4(t, k)$, then 
we can solve {\sc Parameterized} \ProblemName\ in $O(f_4(k, t)^k |V| |E|)$ time
by checking whether $H = (V, E \setminus E')$ is an additive $t$-spanner of $G$ or not
for every subset $E'$ of $F$ with $|E'| = k$.   

Otherwise, we have $|F| \ge f_4(t, k) = N (t+2)^2 f_1(k+t+1, t+1)$.  
Since there exist at least $\frac{|F|}{t+2} \ge N (t+2) f_1(k+t+1, t+1)$ cycles of length at most $t+2$ by the definition of $F$, 
we can take a set $\mathcal C$ of $N(t+2) f_1(k+t+1, t+1)$ cycles of length at most $t+2$.
By Proposition~\ref{prop:12} and Lemmas~\ref{lem:13} and~\ref{lem:14}, 
there always exists a set $E' \subseteq \bigcup_{C\in \mathcal{C}} E(C)$ with $|E'| = k$ such that 
$H=(V, E \setminus E')$ is an additive $t$-spanner of $G$. 
Furthermore, such $E'$ can be found in $O(((t+2) |\mathcal{C}|)^k |V| |E|) = O(f_4(k, t)^k |V| |E|)$ time 
by checking all the edge sets of size $k$ in $\bigcup_{C \in \mathcal{C}} E(C)$. 
Note that it will be possible to improve the running time of this part by following the proofs of 
Proposition~\ref{prop:12} and Lemmas~\ref{lem:13} and~\ref{lem:14}. 
However, we do not do it in this paper, because it does not improve the total running time.

%%Overall, our algorithm solves {\sc Parameterized} \ProblemName\ in $O(f_4(k, t)^k |V| |E|) = (k+t)^{O(kt)} (t+2)^{O(k^2)} |V| |E|$ time, 
Overall, we conclude that our algorithm solves {\sc Parameterized} \ProblemName\ in 
$O(f_4(k, t)^k |V| |E|) = (t+1)^{O(k^2+ t k)} |V| |E|$ time, 
and hence we obtain Theorem~\ref{thm:fpt}. 
The entire algorithm is shown in Algorithm~\ref{alg2}.

	\begin{algorithm}[h]
	\SetKwInOut{Input}{Input}\SetKwInOut{Output}{Output}
	\Input{A graph $G=(V, E)$}
	\Output{An edge set $E' \subseteq E$ with $|E'| = k$ such that 
$H=(V, E \setminus E')$ is an additive $t$-spanner (or conclude that such $E'$ does not exist)}
	Compute $F := \{e \in E \mid \mbox{$e$ is contained in some cycle of length at most $t+2$} \}$  \\
	\If{$|F| \le f_4(t, k)$}{
		\For{each $E' \subseteq F$ with $|E'| = k$}{
			\If{$H = (V, E \setminus E')$ is an additive $t$-spanner of $G$}{
				Return $E'$			
			} 
		}
		Conclude that such $E'$ does not exist
	}
	\Else{ 
		Find a set $\mathcal {C}$ of cycles of length at most $t+2$ with $|\mathcal{C}| = N (t+2) f_1(k+t+1, t+1)$ \\
		\For{each $E' \subseteq \bigcup_{C \in \mathcal C} E(C)$ with $|E'| = k$}{
			\If{$H = (V, E \setminus E')$ is an additive $t$-spanner of $G$}{
				Return $E'$			
			}
		} 
	}
%%
%%	\caption{{\sc Parameterized} \ProblemName}
	\caption{Entire Algorithm}
	\label{alg2}
	\end{algorithm} 

\section{Extension to $(\alpha, \beta)$-Spanners}
\label{sec:extension}

In this section, we extend the argument in the previous section to  $(\alpha, \beta)$-spanners 
and give a proof of Theorem~\ref{thm:fpt2}. 

Let $t := \lfloor \alpha+\beta \rfloor -1$. %% and define $f_4(t, k)$ as in Section~\ref{sec:fpt4}. 
We compute the set $F \subseteq E$ of all edges contained in cycles of length at most $t+2= \lfloor \alpha+\beta \rfloor+1$. 
If $H=(V, E\setminus E')$ is an $(\alpha, \beta)$-spanner of $G$ for $E' \subseteq E$, 
then $\dist_H(u, v) \le \alpha \cdot \dist_G(u, v) + \beta \le \alpha+\beta$ for each $uv \in E'$. 
By integrality, $\dist_H(u, v) \le \lfloor \alpha+\beta \rfloor$ for each $uv \in E'$, 
which shows that $E' \subseteq F$ holds.
This implies that the problem is trivial if $t=0$. 
Thus, we consider the case when $t \ge 1$ and define $f_4(t, k)$ as in Section~\ref{sec:fpt4}.
If $|F| \le f_4(t, k)$, then 
we can solve {\sc Parameterized Minimum $(\alpha, \beta)$-Spanner Problem} in $O(f_4(k, t)^k |V| |E|)$ time
by checking whether $H = (V, E \setminus E')$ is an $(\alpha, \beta)$-spanner of $G$ or not
for every subset $E'$ of $F$ with $|E'| = k$.   

Otherwise, by the argument in Section~\ref{sec:fpt4}, 
in $O(f_4(k, t)^k |V| |E|)$ time, 
we can find an edge set $E'$ with $|E'|=k$
such that $H=(V, E \setminus E')$ is an additive $t$-spanner. 
Then, $H$ is also an $(\alpha, \beta)$-spanner, because
\begin{align*}
\dist_H(u, v) &\le \dist_G(u, v) + t \le (\dist_G(u, v) - 1) + \alpha + \beta  \\ 
                  &\le \alpha \cdot (\dist_G(u, v) - 1) + \alpha + \beta = \alpha \cdot \dist_G(u, v) + \beta
\end{align*}
for every pair of vertices $u$ and $v$. 
Therefore, it suffices to return the obtained set $E'$. 
This completes the proof of Theorem~\ref{thm:fpt2}.

\section{Conclusion}
\label{sec:conclusion}

In this paper, we studied \ProblemName\ from the viewpoint of fixed-parameter tractability. 
We formulated a parameterized version of \ProblemName\ 
in which the number of removed edges is regarded as a parameter, %% in the same way as~\cite{kobayashi2018}, 
and gave a fixed-parameter algorithm for it. 
We also extended our result to {\sc Minimum $(\alpha, \beta)$-Spanner Problem}.  

As described in the last paragraph in Section~\ref{sec:spanner}, 
handling \ProblemName\ is much harder than
{\sc Minimum Multiplicative $t$-Spanner Problem}, 
because we have to care about global properties of graphs. %% when we deal with additive $t$-spanners. 
Since only few results were previously known for \ProblemName, 
this work may be a starting point for further research on the problem. 
%%\ProblemName. 

%%%%{\sc Parameterized Minimum $(\alpha, \beta)$-Spanner Problem} 
%%
%%Finally, we conclude this paper with a remark on the parameter. 
%%In the context of fixed-parameter tractability, several well-studied %%parameterized 
%%problems such as 
%%{Vertex Cover Problem}, {Feedback Vertex Set Problem}, {${\cal F}$-deletion Problem}, and
%%{Multicut Problem} can be described in the following form: 
%%given a graph $G=(V, E)$ and a parameter $k$, 
%%find an edge/vertex subset $Z$ of size {\em at most $k$} %%with $|X| \le k$ 
%%such that $G-Z$ satisfies a certain property, 
%%where the property is closed under {\em removal} of edges/vertices. 
%%In contrast, {\sc Parameterized} \ProblemName\ is not described in this form, 
%%%%contained in this framework, 
%%because 
%%the size of the removed edge subset $F$ is {\em at least} $k$, and 
%%being an additive $t$-spanner is closed under {\em adding} edges. 
%%The study of such a parameterization started recently (see~\cite{Bang2018,kobayashi2018}). 
%%Other 
%%
%%
%%was studied only recently (see~\cite{Bang2018,kobayashi2018}), 
%%and it has a potenti

%%\bibliographystyle{plainurl}

%%\bibliographystyle{plain}
%%\bibliography{spanner}

%%
%% Bibliography
%%

%% Either use bibtex (recommended), 

%% .. or use the thebibliography environment explicitely
\end{document}